\documentclass[a4paper,12pt,english]{article}
\usepackage[utf8]{inputenc}   
\usepackage{lmodern}
\usepackage{cite}
\usepackage[english]{babel}   
\usepackage[babel]{microtype}
\usepackage{amsfonts,amsmath,latexsym,amsthm,fixmath}
\usepackage{amssymb,ifthen}
\usepackage{color,graphicx}
\usepackage{url}
\usepackage{xspace}
\usepackage[
    pdftex,
    pdfstartview={XYZ 0 1000 1.0},
    bookmarks=false,
    colorlinks,
    breaklinks=true,
    citecolor=citecolor,
    filecolor=blue,
    linkcolor=linkcolor,
    urlcolor=urlcolor,
    pdfborder={0 0 0}
]{hyperref}
\definecolor{urlcolor}{rgb}{0, 0.5, 0}
\definecolor{citecolor}{rgb}{.5,0,.25}
\definecolor{linkcolor}{rgb}{0,0,1}

\usepackage{geometry}

\newcounter{withNotes}
\ifthenelse{\value{withNotes} = 0}{
\def\note#1{} }{
\def\note#1{\GenericWarning{}%
  {AUTHOR WARNING: Unresolved annotation}%
    {\marginpar%
      [\hfill\llap{\textcircled{\small{$\circledcirc$}}$\!\Longrightarrow$}]%
      {\rlap{$\Longleftarrow\!$\textcircled{\small{$\circledcirc$}}}}}%
  \textsf{$\langle\!\langle${\red{#1}}$\rangle\!\rangle$}%
}
}

\newcommand{\red}{\textcolor{red}}

\definecolor{OlivierGreen}{rgb}{0.3,0.8,0.4}
\definecolor{Grey}{rgb}{0.7,0.7,0.7}
\definecolor{cyan}{rgb}{0.0, 0.72, 0.92}
\definecolor{Pink}{rgb}{1,0,1}
\definecolor{Red}{rgb}{0.82,0.1,0.26}

\newcommand{\od}[1]{{\color{OlivierGreen}  #1}}

\newtheorem{thm}{Theorem}

\newtheorem{prop}[thm]{Proposition}
\newtheorem{cor}[thm]{Corollary}
\newtheorem{lem}[thm]{Lemma}

\theoremstyle{remark}

\newcommand{\Po}{\ensuremath{P}\xspace}
\newcommand{\Tr}{\ensuremath{\tau}\xspace}
\newcommand{\DT}{\ensuremath{DT(P)}\xspace}
\newcommand{\AT}{\ensuremath{A(\tau)}\xspace}

 \title{Delaunay Triangulations of Points on Circles\thanks{This work
     was supported by the ANR/FNR project SoS,
     INTER/ANR/16/11554412/SoS, ANR-17-CE40-0033.}}

\author{
Vincent Despr\'e
    \thanks{Universit\'e de Lorraine, CNRS, Inria, LORIA, F-54000 Nancy, France.\hfill~\;
 \protect\url{firstname.lastname@inria.fr}}
 \and
Olivier Devillers$^\dag$
 \and
Hugo Parlier%
    \thanks{Mathematics department, University of Luxembourg,
\protect\url{firstname. lastname@uni.lu}}
 \and
Jean-Marc Schlenker$^\ddag$
}

\begin{document}
\maketitle

\begin{abstract}
Delaunay triangulations of a point set in the Euclidean plane are
ubiquitous in a number of computational sciences, including computational geometry.
 Delaunay triangulations are not well
defined as soon as 4 or more points are concyclic but since it is not
a generic situation, this difficulty is usually handled by using a
(symbolic or explicit) perturbation. As an alternative, we propose to
define a canonical triangulation for a set of concyclic points by
using a max-min angle characterization of Delaunay
triangulations. This point of view leads to a well defined and unique
triangulation as long as there are no symmetric quadruples of
points. This unique triangulation can be computed in quasi-linear time
by a very simple algorithm. 
\end{abstract}

\section{Introduction}

Let \Po be a set of points in the Euclidean plane. If we assume that \Po is in general position and in particular do not contain 4 concyclic
 points, then the Delaunay triangulation \DT is the unique
 triangulation over \Po such that the (open) circumdisk of each
 triangle is empty. \DT has a number of interesting properties. The
 one that we focus on is called the {\it max-min angle} property. For
 a given triangulation \Tr, let \AT be the list of all the angles of
 \Tr sorted from smallest to largest. \DT is the triangulation which
 maximizes \AT for the lexicographical
 order~\cite{sibson1978locally,e-acg-87} .
In dimension 2 and for points in general position, this max-min angle property characterizes Delaunay triangulations and highlights one of their most important features: 
they don't contain skinny triangles with small angles.
We call such triangulations max-min angle Delaunay triangulations or
simply Delaunay triangulations as the two notions are equivalent for points in general position.

\smallskip
\noindent\begin{minipage}{0.7\textwidth} 
 We call a quadrilateral $pqrs$ symmetric if there is a symmetry that exchange $p$ with $q$
and $r$ with $s$ or, equivalently, if the two diagonals $pr$ and $qs$ have the same length. 
In such a quadrilateral, the four points are concyclic.
In this paper, we show that considering the max-min angle characterization allows to significantly weaken the notion of general position for \Po. We show the following theorem. 
\end{minipage}\hfill\begin{minipage}{0.25\textwidth}
\includegraphics[width=0.8\textwidth,page=1]{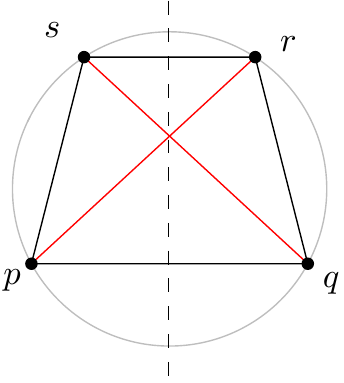}
\end{minipage}

\begin{thm}\label{thm:uniqueness}
If a set of points in the Euclidean plane does not contain any symmetric quadruples then it has a unique max-min angle Delaunay triangulation.
\end{thm}

Notice that the two possible triangulations of a symmetric
  quadrilateral have exactly the same angles (and same diagonal
  length), and thus unicity is impossible when symmetric
  quadrilaterals are allowed.
As an immediate corollary to the theorem above we obtain the following:

\begin{cor}
A set of points in the Euclidean plane with distinct pairwise distances admits a unique max-min angle Delaunay triangulation.
\end{cor}

The usual and generic way to address degeneracies is the use of symbolic perturbations%
~\cite{%
devillers:hal-01586511,
ec-gard-95,
s-nmpgc-98,
y-stgd-90
} 
that perturb the input point set by an infinitesimal quantity in a way that is guaranteed
to remove degeneracies.
This approach has been used for Delaunay triangulation and allows
to draw diagonals in a set of concyclic points in a consistent but not really meaningful manner%
~\cite{ads-rdppw-00,
dt-pdwdt-11}.  
Our result allows to use the max-min Delaunay inside such a set of concyclic points
yielding a meaningful result to triangulate configurations that are
usually considered as degenerate.
 To this aim 
we need an efficient algorithm for this special case.
Such approach has the advantage on symbolic perturbation to define the triangulation
independently of a numbering of the points or of the point coordinates in a particular reference frame.
Notice that in symmetric configurations some cases remain degenerate and symbolic
perturbation cannot help in such a case (unless the perturbation keeps the points concyclic
which seems unpractical).

\begin{thm}\label{thm:complexity}
Fix $n$ points ordered along a Euclidean circle without any symmetric quadruples. 
The unique max-min angle Delaunay triangulation on the $n$ points can be computed using 
$O(n)$ arithmetic operations.
\end{thm}

This note is structured as follows: We give some structural lemmas about concyclic points in Section~\ref{sec:prop}. We prove Theorem~\ref{thm:uniqueness} in Section~\ref{sec:uniqueness}. We then describe our algorithm in a simple setting and prove a weaker version of Theorem~\ref{thm:complexity} in Section~\ref{sec:algo}. We want to present the general idea which is nice and very natural without hiding it in technical details. The details comes in Section~\ref{sec:tech}. Finally, we give an extended algorithm in Section~\ref{sec:extalgo} that can be applied to the most general case.

\section{Properties of concyclic points}\label{sec:prop}

For an integer $n\geq 4$, let $C$ be the unit circle in the
Euclidean plane and \Po$=(p_0,\cdots,p_{n-1})$ be a set of points of $C$
listed in (cyclic) counterclockwise order around $C$. We call {\it
  chords} the segments of the form $[p_ip_{i+1}]$ and {\it diagonals}
the other possible segments. We also call {\it ears} the diagonals of
the form $[p_ip_{i+2}]$. Note that if $n=4$, all segments are ears
or chords, but otherwise there are  diagonals that are not
  ears.  

We are looking for a triangulation $\tau$ of a polygon $P=p_0\cdots p_{n-1}$ which we think of as a decomposition of the polygon into triangles, or alternatively as a collection of edges that cut the polygon into triangles. By an Euler characteristic argument, $\tau$ contains $n-3$ diagonals and $\tau$ cuts $P$ into $n-2$ triangles. The resulting triangles have a total of $3n-6$ angles, all strictly between $0$ and $\pi$. We denote by $A(\tau)=(a_0,\cdots,a_{3n-7})$ the set of these angles listed in increasing order: $a_0\le a_1\le\cdots\le a_{3n-7}$. For different triangulations $\tau$, we order the resulting $A(\tau)$ according to lexicographical order and providing the resulting triangulations with an order. A triangulation $\tau_0$ is said to be {\it angle maximal} if $A(\tau_0) \geq A(\tau)$ for all triangulations $\tau$ of $P$. 

We choose $\tau$ to be one of the triangulations having the list of angles $A(\tau)$  which is maximal for the lexicographical order. Let
 $D(\tau)=(d_0,\cdots,d_{n-3})$ be the list of the diagonals of $\tau$
 such that $\ell(d_0)\le \ell(d_1)\le\cdots\le \ell(d_{n-3})$. For
 simplicity, in the sequel we will not distinguish between a diagonal
 and its length (thus $d_i$ means either the diagonal $d_i$ or its
 length $\ell(d_i)$). As before, for different triangulations $\tau$,
 we order the resulting $D(\tau)$ according to lexicographical
 order. A triangulation $\tau_0$ is said to be {\it length maximal} if
 $D(\tau_0) \geq D(\tau)$ for all triangulations $\tau$ of $P$.  

\begin{lem}\label{lem:diags}
A triangulation of a set of concyclic points $P$ is angle maximal if and only if it is length maximal.
\end{lem}

\begin{proof}
Let $\tau$ be a triangulation of $P$. Each chord $[p_ip_{i+1}]$ is incident to
a triangle of the form $p_ip_{i+1}p_j$. 
All the vertices $p_j$ for $j\neq i, i+1$ lie on the same circular arc of $C$ between $p_{i+1}$ and $p_i$. 

\smallskip
\noindent\begin{minipage}{0.6\textwidth} 
Given a diagonal or a chord $d$ the inscribed angle theorem shows
that the 
angle in the triangle incident to $d$ (on a given side of $d$) at the vertex opposite to $d$
depends only on $d$ and not on the 
position of this vertex on $C$.
Namely, on the side of $d$ that contains the origin this angle is
$\arcsin\frac{d}{2}$ 
and $\pi-\arcsin\frac{d}{2}$ when the origin is on the other side.
Since $\arcsin\frac{d}{2}$ is an 
increasing function of $d$ when $d\in[0,2]$  the angles  
in the triangulation are: 
\end{minipage}\hfill\begin{minipage}{0.3\textwidth}
\includegraphics[width=\textwidth,page=2]{Figures}
\end{minipage}\\
\begin{eqnarray*}
\arcsin\tfrac{d_0}{2}\leq\arcsin\tfrac{d_1}{2}\leq&\ldots&\leq \arcsin\tfrac{d_{n-3}}{2}\\
\leq\pi- \arcsin\tfrac{d_{n-3}}{2}\leq &\ldots&\leq\pi- \arcsin\tfrac{d_{1}}{2}\leq\pi- \arcsin\tfrac{d_{0}}{2}
\end{eqnarray*}
and for a chord $p_ip_{i+1}$
$\arcsin\frac{\|p_ip_{i+1}\|}{2}$ or 
$\pi-\arcsin\frac{\|p_ip_{i+1}\|}{2}$
depending on the side of the origin with respect to $p_ip_{i+1}$. 
Since the angles associated to chords do not depend on a particular triangulation,
they are not relevant when comparing the angles of two triangulations and the above correspondence
between angles associated to diagonals and lengths of these diagonals allows to conclude.
\end{proof}

\noindent\begin{minipage}{0.7\textwidth} 
We will now show that ears are extremal for the lexicographical order of \DT. 
\begin{lem}\label{lem:d0}
For any triangulation of a set of concyclic points, the smallest diagonal, $d_0$, is always an ear.  
\end{lem}
\begin{proof} 

 Let $d$ be a diagonal of the triangulation, then the two other edges of the triangle
incident to $d$ on the side of $d$ that does not contain the origin are shorter
than $d$. Thus if $d$ is the shortest diagonal, these two edges are chords
and $d$ is an ear.
\end{proof}
\end{minipage}\hfill\begin{minipage}{0.25\textwidth}
\includegraphics[width=\textwidth,page=3]{Figures}
\end{minipage}

We denote by $e_i$ the ear $[p_{i-1}p_{i+1}]$.
\smallskip

\noindent\begin{minipage}{0.7\textwidth} 
\begin{lem}\label{lem:triear}
Let $e_i$ and $e_j$ be two non-crossing ears. There exists a triangulation 
containing both $e_i$ and $e_j$.
\end{lem}
\begin{proof}

The triangulation with diagonals
\[e_i,[p_{i-1}p_{i+2}],[p_{i-1}p_{i+3}],\ldots,[p_{i-1}p_{j-2}],[p_{i-1}p_{j-1}]\]
\[e_j,[p_{j-1}p_{j+2}],[p_{j-1}p_{j+3}],\ldots,[p_{j-1}p_{i-2}],[p_{j-1}p_{i-1}]\]
is such a triangulation.
\end{proof}
\end{minipage}\hfill\begin{minipage}{0.25\textwidth}
\includegraphics[width=\textwidth,page=4]{Figures}
\end{minipage}

\smallskip

We can deduce strong structural properties about \DT from the previous lemma.

\begin{prop}\label{prop:path}   
If $P$ is a set of concyclic points, then the dual of \DT is a path. 
\end{prop}
\begin{proof}
We first prove that \DT has at most two ears. Let us consider a triangulation $T_0$ that contains three ears $\{e_i,e_j,e_k\}$ such that $|e_i|\geq |e_j|\geq |e_k|$. Let $T$ be the triangulation given by Lemma~\ref{lem:triear} for the ears $e_i$ and $e_j$. In $T$ all the diagonals different from $e_i$ and $e_j$ have length strictly bigger than $e_j$. Then the list of all the diagonals sorted by length of $T$ has the form $(e_j,d,\cdots)$ or $(e_j,e_i,d,\cdots)$ with $|d|>|e_j|$. However, for $T_0$ we have $(e_k,e_j,\cdots)$ or $(e_k,e_j,e_i,\cdots)$ that is strictly smaller than the list of $T_0$ for the lexicographic order. It means that $T_0$ cannot be \DT and thus that \DT has at most two ears. So, the dual of \DT is a tree with exactly two leaves and and as such is a path.
\end{proof}

Let $E$ be the set of all pairs of disjoint ears of $P$. An element of $E$ is said to be maximal if its shortest ear has maximal length among all elements of $E$. 
Lemma~\ref{lem:triear} implies that, for any maximal element $\{e_i,e_j\}$ of $E$, there exists a triangulation which has as set of ears exactly $e_i$ and $e_j$. This enables us to show the following.

\begin{lem}\label{lem:earDT}
\DT has a maximal element of $E$ as a subset.
\end{lem}
\begin{proof}
Assume that $P$ does not contain four points in symmetric position. 
Let $(e_i,e_j)$ be a maximal couple of $E$ with $e_i$ smaller than $e_j$. 

Lemma~\ref{lem:triear} yields a triangulation whose smallest edge has length  $|e_i|$.
Lemma~\ref{lem:d0} implies that \DT has an ear $e_k$ as smallest
diagonal. Let $e_l$ be another ear of \DT.
On the one hand, comparing the two triangulations,
  Lemma~\ref{lem:diags} gives that $|e_k|\geq |e_i|$
since \DT is length maximal.
On the other hand, comparing the two pair of ears $\{e_i,e_j\}$ and $\{e_k,e_l\}$,
maximality of  $\{e_i,e_j\}$  in $E$ yields $|e_i|\geq |e_k|$.
Thus  $|e_k|= |e_i|$
and the two ears of \DT form also a maximal pair of $E$.
\end{proof}

This lemma is the key of the construction of \DT.
Amongst all ears, the longest one is not a good candidate because it can enforce shorter ear afterwards
while the second (non intersecting) longest ear is always part of a good triangulation.

\section{Uniqueness}\label{sec:uniqueness}

We prove uniqueness in this section.  

\begin{proof}[Proof of Theorem~\ref{thm:uniqueness}] 
\hfill Let \Po be a minimal set of concyclic points admitting two 
\smallskip\\\begin{minipage}{0.75\textwidth} 
distinct   Delaunay 
 triangulations without any symmetric quadruple of points. Let us
assume that there exist 
two disjoint Delaunay
triangulations of \Po, $\Tr_1$ and $\Tr_2$.
By minimality of \Po they cannot share a
diagonal.  By  Lemma~\ref{lem:diags}, 
 there is a length preserving bijective map $b$ between the diagonals
 of $\Tr_1$ and $\Tr_2$. Let $d$ be a diagonal of $\Tr_1$. Note that
 $d$ and $b(d)$ must share a single point, otherwise their endpoints
 form a symmetric quadruple. In addition, if $d'$ is another diagonal
 of $\Tr_1$, then $d$, $b(d)$, $d'$ and $b(d')$ cannot all share the
 same point, again because this would create a symmetric
 quadruple. This implies that each pair is associated to a different
 point. However, the points inside ears of $\Tr_1$ or $\Tr_2$ cannot
 be one of these shared points. And since $\Tr_1$ and $\Tr_2$ have at
 least 2 ears each, those 
\end{minipage}\hfill\begin{minipage}{0.2\textwidth}
\includegraphics[width=\textwidth,page=5]{Figures}
\end{minipage}
 \smallskip\\
edges have to be different since the
 triangulations do not share any diagonals. This implies that at most
 $n-4$ points are the end points of these pairs. This
 contradicts the  
 fact that there are $n-3$ pairs. Hence it is impossible to have two
 disjoint Delaunay triangulations of \Po and this completes the proof.  
\end{proof}

\section{A simplified algorithm}\label{sec:algo}

Lemma~\ref{lem:earDT} suggests an algorithm. We need to find an ear that belongs to all the maximal pairs of $E$. We first describe a simplified version of the algorithm that works in the case where \Po does not admits two diagonals of the same length. This is a stronger condition not having any quadruple of points in symmetric position because it also forbids two diagonals of the same length in the event they share a vertex. The algorithm works as follows.

Consider a set of points \Po. We first compute the three longest ears of \Po. If the two longest ears of \Po are disjoint then we add the second longest to the output triangulation $\Tr_o$. Otherwise we add the third longest edge to $\Tr_o$. Let $i$ be the index of the ear $e_i$ that we just added to $\Tr_o$. Now, we remove $p_i$ from \Po, and proceed inductively until we reach a pentagon where a brute force calculation can easily be done.

To prove that this algorithm has the correct output and to compute
running time, we will need some notation. We denote by $P_k$ the set of
$k$ points obtained after $n-k$ steps of the algorithm where
$n\geq k\geq 5$.
We relabel the remaining points from 0 to $k-1$.
We denote by $e_i^k$ the ear in position $i$ in $P_k$. We denote by $(se_0^k,\cdots,se_{n-k}^k)$ the sorted list of the ears of $P_k$ such that $se_0^k>\cdots>se_{n-k}^k$.

\begin{prop}
The triangulation obtained is the unique Delaunay triangulation: $\Tr_o$=\DT.
\end{prop}
\begin{proof}
We first observe that if there is an ear $e_i$ belonging to $DT(P_k)$,
then the restriction of $DT(P_k)$ to $P_k-p_i$ is $DT(P_k-p_i)$. We
want to show that the chosen ear $e_i^k$ at step $k$ belongs to
$DT(P_k)$. By Lemma~\ref{lem:earDT},  one of the maximal
pair of $E(P_k)$ is in $DT(P_k)$. If $se_0^k$ and $se_1^k$ are
non-crossing then $(se_0^k,se_1^k)$ is the unique maximal pair of
$E(P_k)$. If $se_0^k$ and $se_1^k$ cross then $(se_0^k,se_2^k)$ and
$(se_{\od{1}}^k,se_2^k)$ are the two possible maximal pairs of
$E(P_k)$ and at least one is non crossing.
In the
latter case, $se_2^k$ is the only ear that surely belongs to
$DT(P_k)$. This proves that in all cases, the simplified algorithm
choose an ear that belongs to $DT(P_k)$. 
\end{proof}

Incidentally, this proposition proves the uniqueness of \DT when \Po does not admits two diagonals of the same length using a constructive proof.

\begin{prop}\label{prop:comp}
The simplified algorithm runs using $O(n)$ arithmetic operations.
\end{prop}
\begin{proof}
We first compute the three longest ears of the input polygon on $P_n$. Since there are $n$ ears, it can be computed using $O(n)$ operations. In other words, finding the first ear requires $O(n)$ operations. We want to show that the choices of subsequent ears only require a constant number of operations at each step. We need to update the list of the three longest ears. 

Let $p_i$ be the point of $P_{k+1}$ removed at step $k+1$. Note that $e_i$ cannot contain the origin of the circle since such an ear cannot appear as the smallest ear of a pair.
The ears of $P_k$ are almost the same as the ears of $P_{k+1}$.
Actually, three ears disappear :  $p_{i-2}p_i$, $p_{i-1}p_{i+1}$, and $p_{i}p_{i+2}$ 
and two ears appear:  $p_{i-2}p_{i+1}$ and $p_{i-1}p_{i+2}$.
Since these two ears are longer than the chosen ear $p_{i-1}p_{i+1}$
and at least one of the three longest ear of $P_{k+1}$ remains an ear
of $P_k$ we can guarantee that the three longest ear of $P_k$ must be
chosen in $p_{i-2}p_{i+1}$, $p_{i-1}p_{i+2}$ and the ears remaining
amongst the three longest ear of $P_{k+1}$. Thus selecting these three ears
is done in constant time.
\end{proof}

\section{General Case}\label{sec:tech}

If we allow ears of equal length but no symmetric quadrilaterals,
then two ears of equal length necessarily share a point 
(see Section~\ref{sec:uniqueness}) and thus there at most two of them
(at most three if $n\leq 6$).
 Instead of the three longest ears, we have to use {\it all}
the ears of the three longest possible lengths reachable by ears
$(l_0,l_1,l_2)$. 
We want to apply Lemma~\ref{lem:earDT} to find an ear
that belongs to \DT. We thus study the possible configurations for the
maximal pairs of $E$. The goal is to find an ear that belongs to all
maximal pairs or alternatively find a way to rank those pairs. Let us
start with some easy cases.  
\vspace{1em}
\\\noindent\textbf{Case 1: $|se_0|=|se_1|=l_0$.} 
Here $(se_0,se_1)$ is the unique maximal pair of ears. This implies that $d_0$ and $d_1$ are in \DT.
\vspace{1em}
\\\noindent\textbf{Case 2: $|se_0|=l_0$, $|se_1|=|se_2|=l_1$.} 
Hence $(se_0,se_1)$, $(se_0,se_2)$ and $(se_1,se_2)$ are the only possible maximal pairs.
If $se_0$ crosses both $se_1$ and $se_2$, then $(se_1,se_2)$ is the 
unique maximal pair.
If $se_0$ crosses only $se_1$ (resp. $se_2$) then
 $(se_0,se_2)$ (resp. $(se_0,se_1)$) is the unique maximal pair and
 there is a canonical choice. 
In the remaining case, if we choose to include $(se_0,se_1)$ we can find a triangulation where $|d_0|=l_1$ and $|d_1|>l_1$ using Lemma~\ref{lem:triear}. This implies that $(se_1,se_2)$ cannot be included in \DT since it is strictly worse than any such triangulation. Thus \DT contains $se_0$.
\vspace{1em}
\\The last case occurs when $|se_0|=l_0$, $|se_1|=l_1$ and $|se_2|=|se_3|=l_2$. It will require looking at the next possible steps to decide between maximal pairs.

\begin{lem}\label{lem:tech0}
Let $n>8$. When $|se_0|=l_0$, $|se_1|=l_1$ and $|se_2|=|se_3|=l_2$, we can always find a diagonal that belongs to \DT by considering a constant number of diagonals.
\end{lem}
\begin{proof}
If $se_0$ and $se_1$ are non-crossing then we just choose $se_1$. We \smallskip
now assume \smallskip
\noindent\begin{minipage}{0.35\textwidth} 
	\includegraphics[width=\textwidth,page=6]{Figures.pdf}
\end{minipage}\hfill\begin{minipage}{0.64\textwidth}
that $se_0$ and $se_1$ cross. We choose the indices of the
points such that $se_3=e_1$ and $se_2=e_3$.
Since $n>8$, it is not possible that $se_0$,
$se_1$, $[p_1p_4]$, $[p_2p_5]$, $[p_{-1}p_2]$ and $[p_0p_3]$ contain
three diagonals of the same length. We want to go a step further to
choose the best choice between $se_2$ and $se_3$. The one leading to
the best second diagonal leads to a better $d_1$ with the same $d_0$
and so must be in \DT. 

The next choice of $se_2$ has to be made between $[p_1p_4]$,
$[p_2p_5]$, $se_0$ and $se_1$ because all of the other ears are
shorter than $se_2$. For $se_3$, we have $[p_{-1}p_2]$, $[p_0p_3]$,
\end{minipage}
$se_0$ and $se_1$. Furthermore $[p_1p_4]$, $[p_2p_5]$, $[p_{-1}p_2]$
and $[p_0p_3]$ must have 4 different lengths since we forbid symmetric
quadruples. If there is still an ambiguity, this means that the second
possible diagonal must have the same length in the two cases. In
addition, if the next diagonal chosen in both situations is $se_0$,
then $se_0$ is in \DT. The same holds for $se_1$. It implies that the
only ambiguous case left is when one of $se_0$ and $se_1$ has the same
length as one of $[p_1p_4]$, $[p_2p_5]$, $[p_{-1}p_2]$ or
$[p_0p_3]$. We call $event_i$ the remaining possibilities where we
choose $se_i$ as a first ear. 
\vspace{1em}
\\ \noindent\textbf{Case 1: $se_0=e_4$ and $se_1=e_5$ or $se_0=e_5$
  and $se_1=e_4$.} In $event_2$ the next choice is between $[p_1p_4]$,
$[p_2p_5]$ and $e_5$ since $e_4$ is no longer reachable. Since
$[p_2p_5]$ crosses the two others, the unique maximal pair is
$([p_1p_4],e_5)$. If $[p_1p_4]$ is the smallest then $event_3$ cannot
reach the same length because $[p_1p_4]$ crosses $e_4$ and is smaller
than $e_5$. If $e_5$ is the smallest, then $event_3$ has no other
solution but to include $e_5$ as well (none of its eligible ears have
a common point with $e_5$) and thus $e_5$ is in \DT. If they have the
same length then $e_5$ is also in \DT. 
\vspace{1em}
\\\noindent\textbf{Case 2: $se_0=e_5$ and $se_1=e_6$.} In $event_3$,
$se_0$ and $se_1$ are disjoint from $[p_{-1}p_2]$ and $[p_0p_3]$. This
implies that $[p_{-1}p_2]$ and $[p_0p_3]$ have different lengths than
$se_0$ and $se_1$. So, the only possible ambiguity is if the second
choice is $se_0$ or $se_1$. Since a pair with $se_1$ is beaten by a
pair with $se_0$, $se_0$ needs to be chosen. In $event_2$, we can have
$|se_0|=|p_1p_4|$. However, in this case $(se_0,[p_1p_4])$ is the
unique maximal pair and both ears are selected. So, in this case,
$se_0$ is in \DT. 
\vspace{1em}
\\\noindent\textbf{Case 3: $se_0=e_6$.} For the same reason as in the
previous case, we can assume that $event_3$ selected $se_0$. In
$event_2$, we can have now $|se_0|=|p_2p_5|$. If $(se_0,[p_2p_5])$ is
the only maximal pair we are done. So we need to have
$|p_1p_4|>|p_2p_5|$. In this case, $(se_0,[p_1p_4])$ is the only
maximal pair and $se_0$ is in \DT in both events. 
\vspace{1em}
\\\noindent\textbf{Case 4: $se_0=e_j$ with $j>6$.} $se_1$ cannot appear in any maximal pair since $se_0$ is always better so the only equality holds when both events use $se_0$ and in this case, we include $se_0$ in \DT.

By symmetry between $se_2$ and $se_3$, we have no more cases and this concludes the proof.
\end{proof}

The last case is for $n\le8$. 

\section{Extended Algorithm}\label{sec:extalgo}

\paragraph{Algorithm for triangulations without any symmetric quadruples.}
We want to extend the simplified algorithm of Section~\ref{sec:algo}. First, we see from the last section that we need to consider the four longest ears instead of the three longest. The step that changes is the selection of the ear that we put in \DT at each step. We describe it in detail. Firstly, if $n\le8$, when we have more than one possible ear to put in \DT, we just try all possibilities until one becomes strictly better than all the others. We assume that $n>8$. Let $(se_0,se_1,se_2,se_3)$ be the sorted list of the four longest ears of the current polygon. By considering the cases of the last section, we obtain the following procedure.
\begin{itemize}
\item If $|se_0|>|se_1|>|se_2|>|se_3|$ just apply the same rules as in the simplified version.
\item Else if $|se_0|=|se_1|$, we put any one (or both) in \DT.
\item Else if $|se_0|>|se_1|=|se_2|$ then we put $se_0$ in \DT.
\item Else if $|se_0|>|se_1|>|se_2|=|se_3|$, apply:
\begin{itemize}
\item If $se_0$ and $se_1$ are non-crossing, put $se_1$ in \DT.
\item Else look at the length of the next diagonal coming after the choice of $se_2$ or $se_3$.
\begin{itemize}
\item If one is strictly longer, put the corresponding $se_i$ in \DT.
\item If they are equal, then the length of the second edge must be the length of $se_0$ and $se_1$, so put the corresponding ear in \DT.
\end{itemize}
\end{itemize}
\end{itemize}

This algorithm works by Lemma~\ref{lem:tech0} and runs in $O(n)$ time by Lemma~\ref{lem:tech0} and~\ref{prop:comp}.

\paragraph{Multiple output algorithm.}
We want to have an answer even in degenerate cases. We first
  describe an algorithm that outputs all the optimal
  triangulations. We proceed as follows: if we have many admissible
  ears, then we create as many triangulations as the number of
  admissible ears. We carry on constructing a tree of possibly optimal
  triangulations and we check at each layer which one is the best. If
  there is more than one possibility, we keep them all. We first
  consider the critical case of the regular polygons. 
\begin{lem}\label{l:pol-reg}
If $P$ corresponds to a regular polygon then it admits $n\cdot2^{n-5}$ triangulations as \DT.
\end{lem}
\begin{proof}
Let $P$ be a set of $n$ points forming a regular polygon contained in
a circle of radius 1. We first prove that all the triangulations that
are dual to a path have the same set of lengths and thus are optimal by
Proposition~\ref{prop:path}. Let $T$ be a triangulation whose dual is
a path. A diagonal $d$ have length $f(k)=2\sin(\frac{k\pi}{n})$ where
$k$ is the number of chords on the smallest side of $d$. $f$ is an
increasing function of $k$ since
$k\in[0,\lfloor\frac{n}{2}\rfloor]$. An ear has length $f(1)$ and then
the length of the next diagonal following the dual of $T$ have length
$f(2)$ and so on until we reach $f(\lfloor\frac{n}{2}\rfloor)$ (there
is one or two of this length depending on the parity of $n$) and then
the diagonals length decrease to $f(1)$. So, the sorted set of
diagonal length of $T$ is
$(f(1),f(1),f(2),f(2),\cdots,f(\lfloor\frac{n}{2}\rfloor)(,f(\lfloor\frac{n}{2}\rfloor)))$
and is independent of $T$. 

It remains to count the number of different triangulations dual to a
path. Let us choose a random ear $e$. Then we have to choose the next
diagonal after $e$ and so on until we reach another ear, at each step,
we can choose between two distinct edges. It means that we have $2^{n-4}$
triangulations containing $e$, since a triangulation has $n-3$ edges. There are $n$
distinct possible choices for $e$ but we construct all the
triangulations exactly twice (one for each of its ears). We
obtain $n\cdot2^{n-4}/2=n\cdot2^{n-5}$ distinct triangulations. 
\end{proof}

The complexity of the algorithm is directly linked to the number of
possible \DT and this number is controlled by the number of symmetric
quadruples of $P$. More precisely: 

\begin{prop}
A set of $n$ concyclic points with less than $k$ distinct symmetric quadruples  
admits $h$ different Delaunay triangulations with $h=O(2^{k})$.
These triangulations can be enumerated in $O(nh)$ time.
\end{prop}

\begin{proof}
Let us assume that a decision made by our algorithm involves $k_0
$ ears of the same length. Then, we have $k_0$ different ears of the same size
such that all the pairs are crossing or disjoint 
(but not sharing a point) and so there are necessarily at least
$k_0(k_0-1)$ symmetric quadruples. Now, at each step of the algorithm,
we have a number of equivalent sets of diagonals and we want to find
all the possibles extensions. Since all the sets have to be
equivalent, we must always add ears of the same length to all the
current constructions so $k_0(k_0-1)$ symmetric quadruples may be used
to multiply the number of possible configurations by $k_0$. Until we
reach $n\cdot2^{n-5}$ possibilities the bigger number that we can
obtain is $2^k$ by using each quadruples separately. It proves that
$h=O(2^k)$. 

To actually construct the triangulations, we only need a constant number of operations on each node of the tree
of configurations. A tree as
a linear number of nodes with respect to its number of leaves so the
complexity of the algorithm is $O(nh)$. 
\end{proof}

\paragraph{Single output algorithm.}
For computational applications, it is important to have an
  algorithm that build always the same triangulation on a given set of
  points $P$. Perturbations are not useful because if you have two
  triangulations with the sorted length of diagonals
  $(a,b,\cdots)<(a',c,\cdots)$ with $a=a'$ then a perturbation
   may perturb the lengths $a$ and $a'$ so that $a'> a$, yielding
    an incorrect result.
If the frame of coordinates is fixed, we may use it to pick a
  unique triangulation amongst the optimal ones.
Let $p_0\in P$ be the smallest point for the lexicographic order
and  label the other points $p_i$ in counterclockwise direction
 starting at $p_0$. Now we can choose to set an order on the tree
  constructed by the previous algorithm as follows: if we
  cannot decide between a set of ears $(e_i,e_j,\cdots)$, we order the list
  by label and we put the smallest $e_i$ as the leftmost descendant of
  the previous node. The algorithm outputs the leftmost solution in
  $O(n\cdot2^k)$ time where $k$ is the number of symmetric quadruples
  of $P$. Note that it is not possible to only keep the leftmost
  descendant at each step since it can lead to a non optimal situation. 

\subsection*{Acknowledgements}

The authors thanks François Collet for preliminary discussions about
the problem we addressed in this paper.

\bibliographystyle{plainurl}

\end{document}